\documentclass{article}
\usepackage{amsthm,mathtools,amsmath,amssymb,booktabs}

\def\uball{\mathbb{B}}

\def\reals{\mathbb{R}}

\def\comp{\raise 1pt \hbox{$\scriptstyle\circ$}}

\def\maximize{\mathop{\rm maximize}\limits}
\def\st{\mathop{\rm subject\ to}}

\def\dom{\mathop{\rm dom}}

\def\upto{{\raise 1pt \hbox{$\scriptstyle \,\nearrow\,$}}}
\def\downto{{\raise 1pt \hbox{$\scriptstyle \,\searrow\,$}}}

\def\inte{\mathop{\rm int}\nolimits}

\def\ovr{\mathop{\rm over}}

\newtheorem{theorem}{Theorem}

\newtheorem{lemma}[theorem]{Lemma}
\newtheorem{corollary}[theorem]{Corollary}

\newtheorem{example}[theorem]{Example}
\newtheorem{remark}[theorem]{Remark}
\newtheorem{assumption}{Assumption}

\title{Efficient allocations in double auction markets\footnote{The author is grateful to professor Sjur D.\ Fl{\aa}m for many fruitful discussions and, in particular, for proposing the growth condition in Assumption~\ref{ass:1} to guarantee that the consumer surplus converges to zero.}}
\author{Teemu Pennanen\footnote{Department of Mathematics, King's College London, teemu.pennanen@kcl.ac.uk}
}

\begin{document}

\maketitle

\begin{abstract}
This paper proposes a simple descriptive model of discrete-time double auction markets for divisible assets. As in the classical models of exchange economies, we consider a finite set of agents described by their initial endowments and preferences. Instead of the classical Walrasian-type market models, however, we assume that all trades take place in a centralized double auction where the agents communicate through sealed limit orders for buying and selling. We find that, under nonstrategic bidding, the double auction clears with zero trades precisely when the agents' current holdings are on the Pareto frontier. More interestingly, the double auctions implement Adam Smith's ``invisible hand'' in the sense that, when starting from disequilibrium, repeated double auctions lead to a sequence of allocations that converges to individually rational Pareto allocations.
\end{abstract}


\section{Introduction}

Most modern securities exchanges are based on the double auction mechanism where interested buyers and sellers submit limit orders (a market order can be viewed as a limit order with a very generous limit on the price) and the most generous offers are selected for trade by crossing the supply and demand curves. The same principle is behind day-ahead energy markets where matching takes into account also the transmission capacities and the locations of supply and demand. More recently, double auctions have been implemented for betting and for various crypto currencies.

Ever since the pioneering works of Smith~\cite{smi62}, double auctions have been found to converge quickly to efficient allocations but the phenomenon has remained largely unexplained by theory; see e.g.\ \cite{cf96} and \cite[Section~3]{smi10} or the collections \cite{smi91}, \cite{fr93} and \cite{ps8} for further evidence and analysis. Section~3 of \cite{fri93} surveys mathematical models proposed for the analysis of double auctions. 
In the words of Plott \cite[page~16]{ps8}: ``The tendency of double auction markets to converge to the equilibrium of the associated competitive equilibrium model is well known, but the equilibration process is not understood''.

This paper studies discrete-time double auctions in the classical set-up of welfare economics with a finite set of agents with given endowments and preferences. We assume that all trades take place in a double auction and that the agents submit limit orders according to their indifference (also known as reservation) prices. In general, indifference prices depend not only on the agents' preferences but also on their current endowments which change whenever an agent is involved in a trade. We find that, when the double auction is repeated, the allocations converge to a Pareto allocation that every agent prefers to their original allocation. Moreover, the speed of convergence is linear in the sense that the total consumer surplus is inversely proportional to the number of iterations. 
As essentially proved already by Debreu~\cite{deb54} (see Section~\ref{sec:pareto} below), the surplus is zero if and only if the current allocation is Pareto efficient. The discovered convergence rate thus explains the efficiency of double auctions observed in empirical studies. The obtained convergence result is, to the best of our knowledge, the first analytical justification for the efficiency of the double auction mechanism.

All trading in our model occurs out of equilibrium and the trading stops only at equilibrium. This is not only a feature of our model but also of real markets where disequilibrium is the driving force behind trading; see \cite[Chapter~1]{ps8} or Fisher~\cite{fis83} for a comprehensive discussion and further analysis of disequilibrium economics. Disequilibrium trading can be described also by tatonnement or auction algorithms where, at each iteration, agents can trade arbitrary quantities at given prices which are then updated according to given rule depending on total consumption; see e.g.~\cite[Chapter~2]{fis83} or \cite{gk6,cf2008} for more recent variants with stronger convergence properties. Such algorithms should not, however, be taken as descriptions of real markets where trading costs are nonlinear and price formation is endogenous. One can, on the other hand, view our market model as an algorithm for computing equilibrium allocations. The constructed equilibria are not necessarily Walrasian but merely individually rational allocations on the Pareto frontier. Accordingly, our assumptions on the utilities are weaker and, in particular, do not require the gross substitutes properties often employed in the literature; see e.g.~\cite{gk6,cf2008}.

The advertised convergence occurs when preferences of the agents remain fixed. In practice, however, agents' preferences evolve with the receipt of new information. When the preferences change, so do equilibria and associated equilibrium prices. Combined with a description of how information affects preferences (see Example~\ref{ex:idu} below), our model would give a natural description of how arrival of new information affects trading and market clearing prices in markets based on the double auction mechanism. 

The rest of this paper is organized as follows. The next section starts by reviewing the double auction mechanism as implemented in terms of limit orders in modern securities exchanges. The auction is then formulated in terms the problem of maximizing the consumer surplus. Section~\ref{sec:id} presents our model of the market where a finite set of agents is described by their endowments and preferences over different holdings. Section~\ref{sec:pareto} relates double auction equilibria with Pareto allocations. Section~\ref{sec:ra} proves convergence to Pareto allocations when the double auction is repeated.

\section{Limit orders and double auctions}\label{sec:loda}

Consider a centralized exchange based on the {\em sealed bid double auction} mechanism where a finite set $I$ of agents submit {\em limit orders} to buy and sell a given asset/good. A buy-limit order consists of a {\em price-quantity} pair $(p^b_i,q^b_i)$ where $p^b_i$ is the maximum unit price the agent is willing to pay for the asset and $q^b_i$ is the maximum number of units the agent is willing to buy at this price. Similarly, a sell-limit order $(p^s_i,q^s_i)$ specifies the minimum unit-price and the maximum quantity for selling the asset. An agent interested only in buying or selling would have $q^s_i=0$ or $q^b_i=0$. Submitting both buy and sell orders with nonzero quantities is typical of e.g.\ market makers who provide liquidity to the market. A rational agent would, of course, have $s^b<s^a$ which we will assume throughout. 

At the end of the bidding period, the market is {\em cleared} by matching the maximum quantity $\bar x$ of buy limit orders with sell limit orders. That is, $\bar x$ is the largest number such that $s(\bar x)\le d(\bar x)$, where $s$ and $d$ are the {\em supply} and {\em demand curves}, respectively; see Figure~\ref{fig:mc}. For each $x$, the value $s(x)$ is the {\em marginal price} when buying a total of $x$ units from the most generous sellers. Mathematically, the supply curve is the nondecreasing piecewise constant function given by
\[
s(x) = \inf_{I'\subset I}\{\sup_{i\in I'}p^s_i\,|\,\sum_{i'\in I'}q^s_{i'}\ge x\}.
\]
Analogously, the demand curve is the nonincreasing piecewise constant function given by
\[
d(x) = \sup_{I'\subset I}\{\inf_{i\in I'}p^b_i\,|\,\sum_{i'\in I'}q^b_{i'}\ge x\}.
\]
The $\bar x$ units of the asset are traded at a {\em market clearing price}
\begin{equation}\label{mcp}\tag{MCP}
\bar p\in[s(\bar x),s(\bar x_+)]\cap[d(\bar x),d(\bar x_+)],
\end{equation}
where $s(\bar x_+)$ and $d(\bar x_+)$ denote the right limits of $s$ and $d$, respectively. If either $d$ or $s$ is continuous at $\bar x$, as in Figure~\ref{fig:mc}, then the market clearing price is uniquely defined. If the vertical parts of $d$ and $s$ overlap, there is a whole interval of possible market clearing prices and an additional rule is needed to choose one. A natural choice would be to use the middle point but the conclusions drawn here do not depend on the choice of a market clearing price. All sell orders involved in market clearing trades have limit prices less than or equal to the market clearing price and the involved buy orders have limit prices greater than or equal to the market clearing price. Thus, the agents involved in trading get a price at least as good as the ones they were willing to accept.

\begin{figure}
  \includegraphics[width=\linewidth,height=0.5\linewidth]{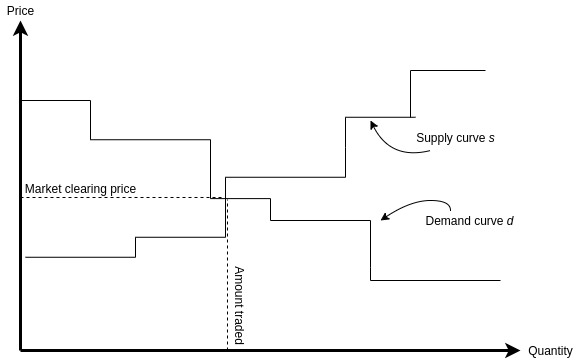}
  \caption{Market clearing in a double auction. }
  \label{fig:mc}
\end{figure}

\begin{figure}
  \includegraphics[width=\linewidth,height=0.5\linewidth]{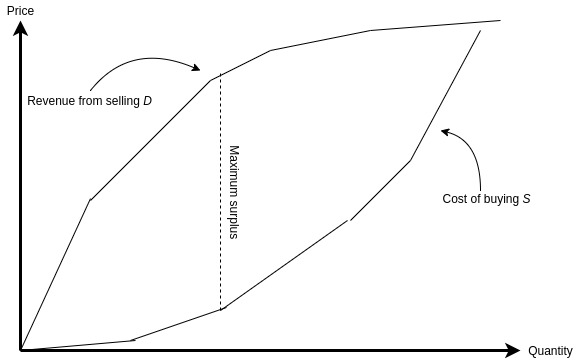}
  \caption{Market clearing maximizes the consumer surplus $D(x)-S(x)$.}
  \label{fig:mcm}
\end{figure}

The double auction mechanism has a variational formulation that will be useful in further analysis. Indeed, the market clearing condition means that $\bar x$ is the largest among all $x\ge 0$ that maximize the ``consumer surplus'' $D(x)-S(x)$, where $S(x)$ is the {\em least cost} of buying $x$ units from the potential sellers and $D(x)$ is the {\em greatest revenue} one could get by selling $x$ units to potential buyers; see Figure~\ref{fig:mcm}. Mathematically, 
\[
S(x) = \inf_{x_i}\left\{\sum_{i\in I}S_i(x_i)\,\right|\left.\sum_{i\in I}x_i=x\right\},
\]
where
\[
S_i(x_i)=
\begin{cases}
  p^s_ix_i & \text{if $x_i\in[0,q^s_i]$},\\
  +\infty & \text{otherwise}
\end{cases}
\]
is the amount of cash agent $i$ would require for selling $x_i$ units of the asset.
Analogously,
\[
D(x) = \sup_{x_i}\left\{\sum_{i\in I}D_i(x_i)\,\right|\left.\sum_{i\in I}x_i=x\right\},
\]
where
\[
D_i(x_i)=
\begin{cases}
  p^b_ix_i & \text{if $x_i\in[0,q^b_i]$},\\
  -\infty & \text{otherwise}
\end{cases}
\]
is the amount of cash agent $i$ would be willing to pay for $x_i$ units of the asset.

The functions $S$ and $D$ can be expressed also as the indefinite integrals of the supply curve and demand curves, $s$ and $d$, respectively. Conversely, the demand and supply curves $d$ and $s$ can recovered from the functions $D$ and $S$ through the relations
\[
\partial D(x)=[d(x),d(x+)]\quad\text{and}\quad\partial S(x)=[s(x),s(x+)],
\]
where $\partial D(x)$ is the {\em superdifferential} of $D$ at $x$, i.e.\ the set of prices $p$ with
\[
D(x)\le D(\bar x)+p(x-\bar x)\quad\forall x\ge 0.
\]
Similarly, $\partial S(\bar x)$ is the {\em subdifferential} of $S$ at $\bar x$, i.e.\ the set of prices $p$ with
\[
S(x)\ge S(\bar x)+p(x-\bar x)\quad\forall x\ge 0.
\]
We refer the reader to \cite[Section~23]{roc70a} for a general study of sub- and superdifferentials in finite-dimensional spaces. 

Condition \eqref{mcp} can now be written as
\[
p\in\partial D(\bar x)\cap\partial S(\bar x).
\]
This implies that $0\in\partial[D-S](\bar x)$ which means that $\bar x$ maximizes the consumer surplus $D(x)-S(x)$ as claimed earlier. Indeed, by definition of superdifferential, $0\in\partial[D-S](\bar x)$ means that
\[
D(x)-S(x)\le D(\bar x)-S(\bar x) + 0(x-\bar x)\quad \forall x\ge 0.
\]
Besides this standard formulation, there is another variational formulation of the double auction mechanism that turns out to be useful when studying its efficiency. Indeed, plugging in the definitions of $D$ and $S$, we can write the surplus maximization problem as
\[
\begin{aligned}
  &\maximize\quad & &\sum_{i\in I}D_i(x^+_i) - \sum_{i\in I}S_i(x^-_i) \quad\ovr\ x^+,x^-\in\reals^I\\
  &\st\quad & & \sum_{i\in I}x^+_i=\sum_{i\in I}x^-_i,
\end{aligned}
\]
where $x^+_i$ and $x^-_i$ denote the purchases and sales, respectively, of agent $i$. Note that since $p_i^b<p_i^s$, one has either $x^+_i=0$ or $x^-_i=0$. Interpreting negative purchases as sales and extending the definition of $D_i$ by
\begin{align*}
D_i(x_i)&:=\sup_{x^+_i,x^-_i}\{D_i(x^+_i)-S_i(x^-_i)\,|\,x^+_i-x^-_i=x_i\}\\
&=\begin{cases}
p^b_ix_i & \text{if $x_i\in[0,q_i^b]$},\\
p^s_ix_i & \text{if $x_i\in[-q_i^s,0]$},\\
-\infty & \text{if $x_i\notin[-q_i^s,q_i^b]$},
\end{cases}
\end{align*}
we can write the market clearing problem more concisely as
\begin{equation}\label{ps}\tag{$\tilde P$}
\begin{aligned}
  &\maximize\quad & &\sum_{i\in I}D_i(x_i) \quad\ovr\ x\in\reals^I\\
  &\st\quad & & \sum_{i\in I}x_i=0.
\end{aligned}
\end{equation}
The functions $D_i$ are {\em concave} since $p_i^b<p_i^s$. They contain exactly the same information as the agents' limit orders. In auction theory, $D_i(x_i)$ is known as the ``reservation price'' for $x_i$. It gives the maximum amount of cash agent $i$ is prepared to pay for $x_i$ units of the asset. Again, buying negative quantities is interpreted as sales and negative payments are income. Problem~\eqref{ps} can be interpreted as the maximum revenue the auctioneer could generate by buying the asset from some agents and selling to others at their reservation prices. This interpretation is not practically relevant, however, as the exchange simply implements the double auction and is not involved with trading. The agents involved in market clearing trade at the market clearing price.

Submitting several buy and sell orders, an agent can effectively submit any {\em concave function} $D_i$ to the exchange. Indeed, if agent $i$ submits finite collections $(p_{i,k}^b,q_{i,k}^b)_{k\in K^b}$ and $(p_{i,k}^s,q_{i,k}^s)_{k\in K^s}$ of buy and sell limit orders (here, $K^b$ and $K^s$ denote finite sets of buy and sell orders, respectively), we obtain the same market clearing problem \eqref{ps} but now
\begin{align*}
D_i(x_i)&=\sup_{x^+_{i,k},x^-_{i,k}}\left\{\left.\sum_{k\in K^b}D_{i,k}(x^+_{i.k})-\sum_{k\in K^s}S_{i,k}(x^-_{i,k})\,\right|\,\sum_{k\in K^b}x^+_{i,k}-\sum_{k\in K^s}x^-_{i,k}=x_i\right\}.
\end{align*}
Indeed, all the submitted limit orders are binding agreements to buy/sell at the offered prices up to the specified quantities. The function $D_i$ can be interpreted as the greatest amount of cash the auctioneer could get by selling $x_i$ units of the asset to agent $i$. Any piecewise linear concave function can be expressed in this form. Any concave function can, in turn, be approximated by piecewise linear functions to arbitrary accuracy

\section{Price-taking agents in multi-asset auctions}\label{sec:id}

Existing double auctions reviewed in the previous section only involve two assets: the asset being auctioned and cash. Our subsequent analysis will allow for multiple assets as it turns out that this does not present any complications in theory. It will, however, allow us to make comparisons with classical welfare economics and, in particular, general equilibrium models. The usual situation with two traded assets is covered as a special case.

Consider an economy with a finite set $J$ of assets and assume that agent $i\in I$ has initial endowment $x_i^0\in\reals^J$. Preferences of agent $i$ over different portfolios of assets are measured by a utility function $u_i$ on $\reals^J$. The setting is thus similar to that of pure exchange models of classical welfare economics; see e.g.\ \cite[Chapter~16]{mwg95} or \cite[Chapter~17]{var92}. We deviate, however, from the classical Walrasian-type market models where trading occurs at exogenously given equilibrium prices. Instead, we assume that all trades are executed in a call auction where prices and participating offers are determined according to the double auction mechanism. In general, a single call auction will not lead to an equilibrium but, as we will see in Section~\ref{sec:ra}, repeating the call auction from the updated positions gives a sequence of allocations that converges to an equilibrium.

Following the formulation \eqref{ps} of the double auction mechanism in the two-asset case, we assume that the agents submit concave functions $D_i$ specifying how much they would be willing to pay for a given portfolio of assets. The market is then cleared according to the double auction mechanism which, in the present multi-asset setting, amounts to solving the optimization problem
\begin{equation}\label{p}\tag{$P$}
  \begin{aligned}
    &\maximize\quad & & \sum_{i\in I}D_i(x_i)\quad\ovr\ x\in\reals^{I\times J}\\
    &\st\quad & & \sum_{i\in I}x_i=0
  \end{aligned}
\end{equation}
which is a straightforward generalization of the market clearing problem \eqref{ps} for a single asset. While in the single-asset setting, general concave functions $D_i$ can be approximated by submitting a collection of limit orders, problem \eqref{p}  should be taken as theoretical abstraction of limit order trading in exchanges implemented today. To submit approximations of concave functions in the multi-asset setting would require the development of exchanges where limit orders and market clearing involve multiple assets.

We will assume that each agent defines the value $D_i(x_i)$ as the maximum amount she could pay for a portfolio $x_i$ without worsening her current utility. More precisely, we assume {\em nonstrategic bidding} in the sense that each agent bids according to her {\em indifference prices} (reservation prices, willingness to pay, \ldots) defined by
\begin{equation}\label{idp}
D_i(x_i) := \sup\{r\in\reals\,|\, u_i(x^0_i+x_i-r g)\ge u_i(x^0_i)\},
\end{equation}
where $g$ denotes a {\em numeraire portfolio} in terms of which all prices are quoted. If prices are quoted in terms of cash, $g$ is the unit vector with nonzero component only for the cash asset; see Example~\ref{ex:cash} below. The value of $D_i(x_i)$ is the greatest amount of cash agent $i$ would be willing to pay for a portfolio $x_i$. Paying more would reduce the agent's utility. Again, negative purchases are interpreted as sales so that $S_i(x_i)=-D_i(-x_i)$ gives the indifference price for selling $x_i$. Thus, both buyers and sellers can be described by the indifference functions $D_i$.

Bidding the indifference price function is rational in one-off double auctions that are large enough so that an individual agent's bid has a negligible effect on the market clearing price. An agent's bid then only affects the quantity she buys while every unit bought below her indifference price increases her utility. When considering repeated double auctions as in Section~\ref{sec:ra} below, the situation becomes more complicated as one may want to postpone trading in the hope of beneficial market price developments. 
Roberts and Postlewaite~\cite{rp76} assume a general exchange mechanism and show that, possible gains from non-competitive behavior goes to zero as the number of agents increases. Our assumption of indifference pricing is akin to that made in Friedman~\cite{fri91} who assumed that the agents neglect strategic feedback effects and bid according to abstract strategies satisfying certain plausibility assumptions. Such assumptions are supported by Cason and Friedman~\cite{cf93} who find that relatively simple bidding rules explain human behavior in double auctions better than more sophisticated strategies. Plott~\cite[Chapter~10]{ps8} studies three two-sided auction mechanisms and finds that individuals tend to gain little from strategic bidding. Friedman~\cite[page~30]{fri10} summarizes the findings of a series of empirical studies as ``\ldots traders in simple static theoretical models have the incentive to understate their willingness to transact, and those understatements lead to inefficient outcomes in thin markets. However, stationary repetition in the CDA teaches traders not to understate inefficiently: they learn to shade bids and asks towards earlier transaction prices, but not beyond. The learning process leads to a Nash equilibrium that implements a competitive equilibrium outcome''.

We allow for extended real-valued and nondifferentiable utility functions $u_i$ but assume the following throughout.

\begin{assumption}\label{ass:g}
The utility functions $u_i$ are upper semicontinuous, concave and strictly increasing in the direction of the numeraire asset $g$ in the sense that
\[
u_i(x_i+rg)>u_i(x_i)
\]
for all $x_i\in\dom u_i$ and $r>0$.
\end{assumption}

Here and in what follows,
\[
\dom u_i := \{x_i\in\reals^J\mid u_i(x_i)>-\infty\}.
\]
Extended real-valued utility functions allow for incorporation of constraints with infinite penalties. This is essential e.g.\ with Cobb-Douglas-type utility functions (see Section~\ref{sec:num}) whose domains are the positive orthants $\reals^J_+$. In the general setup, however, we do not require $\dom u_i\subseteq\reals^J_+$, so short positions may be feasible. Also, besides growth in the direction of $g$, we do not assume any monotonicity properties so some agents may find utility in assets that others find undesirable altogether. Since we do not require differentiability, our analysis covers also Leontief-type utility functions
\[
u_i(x_i)=\min_{j\in J}\{\alpha_{i,j}x_{i,j}\},
\]
where $\alpha_{i,j}$ are real parameters. Nonsmoothness is essential also in describing agents with production facilities.

\begin{example}[Indirect utility functions]\label{ex:idu}
Our assumptions allow for indirect utility functions of the form
\[
u_i(x_i) = \sup_{y_i}\{U_i(y_i)\,|\,y_i\in Y_i(x_i)\},
\]
where $Y_i$ is a von Neumann--Gale-type ``production mapping'' describing how a vector $x_i$ of goods can be transformed into another $y_i$ and $U_i$ is a utility function on the outputs. As soon as the set $G:=\{(x_i,y_i)\mid y_i\in Y_i(x_i)\}$ is convex and $U_i$ is concave, the function $u_i$ will be concave as well. Applying \cite[Theorem~9.2]{roc70a} would give quite general sufficient conditions for the upper semicontinuity. A sufficient condition is that $U_i$ is upper semicontinuous, $G$ is closed and $Y_i(0)$ be bounded.

One could also allow for random production mappings and utilities and define
\[
u_i(x_i)=\sup_{y_i}\{E^{P_i}U_i(y_i)\,|\,y_i\in Y_i(x_i)\ P_i\text{-a.s.}\},
\]
where the supremum is taken over random outputs $y_i$ and $P_i$ is a probability measure describing agent $i$'s views about the uncertain future states. Such indirect utility functions incorporate agents' information through the subjective probability measures $P_i$ which may evolve at the arrival of news and other information. Changes in the subjective probabilities would affect the indirect utilities $u_i$ and thus, market clearing prices, as we will see below. Indirect utilities are often often used in indifference pricing in incomplete market models of financial economics; see e.g.\ \cite[Chapter~5]{car9} or \cite{pen14} and the references there.
\end{example}

In what follows, we assume that the function $D_i$ submitted by agent $i$ is defined by \eqref{idp}. The following is a simple consequence of Assumption~\ref{ass:g}; see e.g.\ \cite[Proposition~2]{pen12}.

\begin{lemma}\label{lem:usc}
Under Assumption~\ref{ass:g}, the function $D_i$ is concave and upper semicontinuous with $D_i(0)=0$ and
\[
D_i(x_i+rg)=D_i(x_i)+r
\]
for all $x_i\in\dom D_i$ and $r\in\reals$.  
\end{lemma}

\begin{example}[Cash as numeraire]\label{ex:cash}
In many markets, the numeraire asset is cash. If cash is denoted by $0\in J$, this means that the numeraire $g$ in the definition \eqref{idp} of the function $D_i$ is a unit vector with $g_0=1$ and $g_j=0$ for $j\in\tilde J:=J\setminus\{0\}$. Denoting $x_i=(x_{i,0},\tilde x_i)\in\reals\times\reals^{\tilde J}$, Lemma~\ref{lem:usc} gives
\[
D_i(x_i) = D_i(0,\tilde x_i) + x_{i,0}
\]
so problem \eqref{p} can be written as
\[
\begin{aligned}
  &\maximize\quad & & \sum_{i\in I}D_i(0,\tilde x_i)+\sum_{i\in I}x_{i,0}\quad\ovr\ x\in\reals^{I\times J}\\
  &\st\quad & & \sum_{i\in I} x_{i,0}=0,\\
  & & & \sum_{i\in I}\tilde x_i=0,
\end{aligned}
\]
or equivalently,
\[
\begin{aligned}
  &\maximize\quad & & \sum_{i\in I}D_i(0,\tilde x_i) \quad\ovr\ \tilde x\in\reals^{I\times\tilde J}\\
  &\st\quad & & \sum_{i\in I}\tilde x_i=0.
\end{aligned}
\]
When $\tilde J$ is a singleton, we recover problem \eqref{ps} in the two-asset setting of Section~\ref{sec:loda} where a single auctioned asset is paid for in cash.
\end{example}

The market clearing prices will be the {\em Lagrange multipliers} associated with the market clearing constraint in \eqref{p}. In order to guarantee the existence of market clearing prices, we will consider the following {\em perturbed} market clearing problem
\begin{equation}\label{pz}\tag{$P_z$}
  \begin{aligned}
    &\maximize\quad & & \sum_{i\in I}D_i(x_i)\quad\ovr\ x\in\reals^{I\times J}\\
    &\st\quad & & \sum_{i\in I}x_i=z,
  \end{aligned}
\end{equation}
where $z\in\reals^J$. The optimum value can be interpreted as the maximum revenue the auctioneer could get by selling the portfolio $z$ to the market participants at their indifference prices. This interpretation is, however, irrelevant as the auctioneer simply implements the double auction mechanism and does not get involved with trading otherwise. When $z=0$, problem \eqref{pz} becomes the market clearing problem \eqref{p}. The existence of a Lagrange multiplier is equivalent to the subdifferentiability at the origin of the optimum value of \eqref{pz} as a function of $z$; see \cite[Theorem~16]{roc74}. Since the optimum value function is concave, the subdifferentiability is implied by continuity at the origin; see \cite[Theorem~11]{roc74}. By \cite[Theorem~8]{roc74}, the following implies the continuity. 

\begin{assumption}\label{slater}
There is an $\varepsilon>0$ such that the optimum value of \eqref{pz} is finite for all $z\in\reals^J$ with $|z|\le\varepsilon$.
\end{assumption}

A sufficient condition for Assumption~\ref{slater} is that, for each asset $j\in J$, there are agents $i,i'\in I$ such that $D_i(e_j)$ and $D_{i'}(-e_j)$ are finite. Here $e_j$ denotes the vector with $\epsilon$ at the $j$th component and zeros elsewhere. Indeed, we then have that the optimum value is finite at $e_j$ and $-e_j$ for all $j\in J$. By concavity, the optimum value is then finite over the convex hull of such vectors, a set which contains a neighborhood of the origin. In the two-asset setting of Section~\ref{sec:loda}, this holds if there is at least one seller and one buyer. Assumption~\ref{slater} could be weakened to requiring that the origin belongs to the relative interior of the domain of the optimum value function.

The following is a simple application of the classical optimality conditions in convex optimization; see e.g.\ \cite[Theorem~15]{roc74} or \cite[Section~28]{roc70a}.

\begin{theorem}\label{kkt}
Under Assumption \ref{slater}, an $\bar x$ solves \eqref{p} if and only if there exists a price vector $p\in\reals^J$ such that
\begin{align}
  \partial D_i(\bar x_i) &\ni p\quad\forall i\in I,\label{kkt1}\\
  \sum_{i\in I}\bar x_i&=0.\label{kkt2}
\end{align}
\end{theorem}

Condition \eqref{kkt1} means that $p$ is a {\em supergradient} of $D_i$ at $\bar x_i$, i.e.
\begin{equation}\label{sg}
D_i(x_i)\le D_i(\bar x_i)+p\cdot (x_i-\bar x_i)\quad\forall x_i\in\reals^J.
\end{equation}
The vector $p$ in Theorem~\ref{kkt} is a {\em market clearing price}: agent $i$ will pay $p\cdot\bar x_i$ units of the numeraire portfolio $g$ for $\bar x_i$. Thus, at market clearing, the agent's portfolio will be updated to
\[
x_i^1 := x^0_i + \bar x_i-(p\cdot\bar x_i)g.
\]
Again, if a component $x_{i,j}$ of $x_i$ is negative, then agent $i$ is selling asset $j$ and receiving $-p_j\bar x_{i,j}$ units of cash for it. The market clearing condition \eqref{kkt2} implies
\begin{equation}\label{eq:mcc}
\sum_{i\in I}x^1_i = \sum_{i\in I}x^0_i
\end{equation}
so the new allocation is feasible. Choosing $x_i=\bar x_i+rg$ in \eqref{sg} and using the last property in Lemma~\ref{lem:usc} gives
\[
g\cdot p = 1,
\]
i.e.\ the numeraire is priced at $1$. In the setting of Example~\ref{ex:cash} where $g$ is a unit vector, this just means that $p_0=1$ (one unit of cash is worth one unit of cash).

\begin{remark}[Consumer surplus]\label{rem:cs}
Condition \eqref{kkt1} or, equivalently, \eqref{sg} implies
\[
D_i(0)\le D_i(\bar x_i)-p\cdot\bar x_i
\]
where, by Lemma~\ref{lem:usc}, the left hand side equals zero under Assumption~\ref{ass:g}. Thus, the payment $p\cdot\bar x_i$ is less than $D_i(\bar x_i)$ which is what the agent was prepared to pay for $\bar x_i$. The difference may be thought of as agent $i$'s {\em consumer surplus}. Any feasible allocation $x$ satisfies
\[
\sum_{i\in I}D_i(x_i) = \sum_{i\in I}[D_i(x_i)-p\cdot x_i]
\]
so the double auction mechanism maximizes the {\em total consumer surplus} over all feasible allocations. 
\end{remark}  

\begin{remark}[Pareto improvements]\label{rem:paretoimp}
The double auction mechanism makes a Pareto improvement of allocations. Indeed, we have
\[
x_i^1 = x^0_i + \bar x_i-D_i(\bar x_i)g + [D_i(\bar x_i)-p\cdot\bar x_i]g.
\]
Under Assumption~\ref{ass:g}, $D_i(\bar x_i)-p\cdot\bar x_i\ge 0$ and
\[
u_i(x^1_i)\ge u_i(x^0_i + \bar x_i-D_i(\bar x_i)g)
\]
while $u_i(x^0_i + \bar x_i-D_i(\bar x_i)g)\ge u_i(x^0_i)$, by the definition of the indifference price $D_i(\bar x_i)$. Thus,
\[
u_i(x^1_i)\ge u_i(x^0_i)
\]
where the inequality is strict unless $D_i(\bar x_i)=p\cdot\bar x_i$. The equality would mean that agent $i$ submitted a buy order with limit price equal to the market clearing price.
\end{remark}  

\begin{remark}[Multiplicity of market clearing prices]
Just as in the single asset auction, the market clearing prices $p$ need not be unique. The conclusions drawn here do not depend on the choice but, in practice, of course, the choice is important to all agents involved in market clearing trades.
\end{remark}  

\begin{remark}[Negative prices]
The market clearing constraint in problem \eqref{p} is an equality instead of an inequality. This means that we do not allow for free disposal of the assets. It follows that some of the market clearing prices may be strictly negative. This has practical significance as some of the assets may become a liability to some agents. Examples of the phenomenon have been observed e.g.\ in electricity markets where excess supply during periods of high windpower production has been met with low demand and significantly negative electricity prices. Another example is the negative money market rates in the eurozone since June 2014.
\end{remark}

\section{Existence of solutions}

The double auction mechanism discussed above, only makes sense if problem \eqref{p} admits optimal solutions $\bar x$. In the two-asset setting of Section~\ref{sec:loda}, a solution exists except in the unrealistic case where the demand curve lies strictly above the supply curve on the whole positive axis. This would mean that an infinite quantity of buy and sell offers could be matched. Theorem~\ref{thm:exist} below gives sufficient conditions for existence in the general case.

The market clearing problem \eqref{p} is written in terms of the functions $D_i$ defined through minimization in \eqref{idp}. Plugging in the definitions, we can write the problem in terms of the utility functions $u_i$ as follows.

\begin{lemma}\label{lem:opt}
Problem \eqref{p} is equivalent to the problem
\begin{equation}\label{opt}\tag{$P'$}
\begin{aligned}
&\maximize\quad & & r\quad\ovr\ r\in\reals,\ w\in\reals^{I\times J}\\
  &\st\quad & & \sum_{i\in I}w_i+rg=\sum_{i\in I}x^0_i,\\
  & & & u_i(w_i) \ge u_i(x_i^0)\quad i\in I
\end{aligned}
\end{equation}
in the sense that their optimum values coincide and an $\bar x$ solves \eqref{p} if and only if there exist $r_i$ such that
\[
r = \sum_{i\in I}r_i
\]
and $w_i=x^0_i+\bar x_i - r_ig$ solve \eqref{opt}.
\end{lemma}

\begin{proof}
Using the definition of $D_i(x_i)$, we can write problem \eqref{p} as
\begin{equation*}
\begin{aligned}
&\maximize\quad & & \sum_{i\in I}r_i\quad\ovr\ r\in\reals^I,\ x\in\reals^{I\times J}\\
  &\st\quad & &\sum_{i\in I}x_i=0,\\
  & & & u_i(x_i^0+x_i-r_i g) \ge u_i(x_i^0)\quad i\in I,
\end{aligned}
\end{equation*}
or in terms of $w_i:=x^0_i+x_i-r_i g$, as
\begin{equation*}
\begin{aligned}
&\maximize\quad & & \sum_{i\in I}r_i\quad\ovr\ r\in\reals^I,\ w\in\reals^{I\times J}\\
  &\st\quad & & \sum_{i\in I}w_i + \sum_{i\in I}r_ig=\sum_{i\in I}x^0_i,\\
  & & & u_i(w_i) \ge u_i(x_i^0)\quad i\in I.
\end{aligned}
\end{equation*}
This is the problem in the statement with $r=\sum_{i\in I}r_i$.
\end{proof}

By Lemma~\ref{lem:opt}, the market clearing problem \eqref{p} has a solution if and only if problem \eqref{opt} has one. Theorem~\ref{thm:exist} below gives sufficient conditions for existence in terms of the {\em recession functions} defined for the upper semicontinuous concave functions $u_i$ by
\[
u_i^\infty(x_i):=\inf_{\alpha>0}\frac{u_i(\bar x_i + \alpha x_i)-u_i(\bar x_i)}{\alpha},
\]
where $\bar x_i\in\dom u_i$. By \cite[Theorem~8.5]{roc70a}, the definition is independent of the choice of $\bar x_i\in\dom u_i$. The recession function describes the asymptotic behavior of $u_i$ infinitely far from the origin. It is concave and positively homogeneous. If $u_i$ is positively homogeneous, then $u_i^\infty=u_i$. 

\begin{assumption}\label{ass:rec}
If $x\in\reals^{I\times J}$ is such that
\[
u_i^\infty(x_i)\ge 0,\ \sum_{i\in I}x_i=0,
\]
then $x=0$.
\end{assumption}

Assumption~\ref{ass:rec} holds if there is a pointed convex cone $K$ that contains, for each $i\in I$, the ``directions of recession''
\[
\{x\in\reals^J\mid u_i^\infty(x)\ge 0\}
\]
of $u_i$; see \cite{roc70a}. By \cite[Theorem~8.6]{roc70a}, directions of recession are precisely the vectors $x$ such that $\lambda\mapsto u_i(\bar x+\lambda x)$ is nondecreasing for all $\bar x\in\dom u_i$. For Cobb-Douglas utilities, one can take $K=\reals^J_+$. The same works for Leontief utilities under the usual assumption of strictly positive parameters $\alpha_{i,j}$.

\begin{theorem}\label{thm:exist}
The market clearing problem \eqref{p} admits solutions under Assumptions~\ref{ass:g} and \ref{ass:rec}.
\end{theorem}

\begin{proof}
It suffices to prove that problem \eqref{opt} in Lemma~\ref{lem:opt} has a solution. By \cite[Corollary~27.3.3]{roc70a}, it suffices to show that the set
\[
\{(r,w)\in\reals\times\reals^{I\times J}\mid r\ge 0,\ \sum_{i\in I}w_i+rg=0,\ u_i^\infty(w_i)\ge 0\}
\]
only contains the origin. Under Assumption~\ref{ass:g}, any element $(r,w)$ of this set satisfies
\[
r\ge 0,\ \sum_{i\in I}(w_i+\frac{r}{|I|}g)=0,\ u_i^\infty(w_i+\frac{r}{|I|}g)\ge 0
\]
and, by Assumption~\ref{ass:rec}, $w_i+\frac{r}{|I|}g=0$. Since $u_i(w_i)\ge 0$, this implies $u_i^\infty(-\frac{r}{|I|}g)\ge 0$ which, under Assumption~\ref{ass:g}, can only hold if $r=0$ and then, $w_i=0$ as well.
\end{proof}




\section{Double auction equilibria and Pareto allocations}\label{sec:pareto}

By Remark~\ref{rem:cs}, the optimum value of the market clearing problem \eqref{p} equals the maximum consumer surplus over all feasible reallocations. In order to emphasize its dependence on the current allocation $x^0$, we will denote it by $CS(x^0)$. 
Since $x=0$ is feasible in \eqref{p}, and since, by Lemma~\ref{lem:usc}, $D_i(0)=0$ under Assumption~\ref{ass:g}, we have $CS(x^0)\ge 0$. If $CS(x^0)=0$, the market clearing problem \eqref{p} is solved by $x=0$, and we say that $x^0$ is a {\em double auction equilibrium}. In other words, double auction equilibria are allocations $x^0$ at which the market clears with zero trades and all agents have zero surplus\footnote{As the solution to the market clearing problem need not be unique, there may be other solutions besides $0$ at a double auction equilibrium $x^0$. In the two-asset setting of Section~\ref{sec:loda}, this would mean that horizontal parts of the supply and demand curves $s$ and $d$ overlap. The uniqueness or the lack of it has no effect on the conclusions of this paper.}.

It is natural to ask, how ``efficient'' are double auction equilibria. They are, after all, defined in terms of a specific market mechanism, the double auction. Much like the fundamental theorems of welfare economics relate Walrasian equilibria with Pareto efficient allocations (see e.g.\ \cite[Chapter~16]{mwg95} or \cite[Chapter~17]{var92}), the main result of this section, Theorem~\ref{thm:pareto} below, gives conditions under which double auction equilibria coincide with Pareto allocations. Recall that a feasible allocation $x$ is {\em Pareto efficient} if there does not exist another feasible allocation $x'$ such that
\[
u_i(x'_i)\ge u_i(x_i)
\]
for all agents $i\in I$ with strict inequality for at least one of them.

Lemma~\ref{lem:opt} allows for a quick proof of the fact that, under Assumption~\ref{ass:g}, Pareto allocations are double auction equilibria. We will prove the converse under the following.

\begin{assumption}\label{ass:0}
If $x_i$ is such that $u_i(x_i)>u_i(x_i')$ for some $x'_i\in\dom u_i$ then $x_i-\varepsilon g\in\dom u_i$ for small enough $\varepsilon>0$.
\end{assumption}

Assumption~\ref{ass:0} means that starting from a point which is not of lowest possible utility over $\dom u_i$, the agent can give away some positive amount of the numeraire portfolio without making his position infeasible. This clearly holds if the sets $\{x_i\in\reals^J\,|\, u_i(x_i)>u_i(x_i')\}$ with $x_i'\in\dom u_i$ are open, or if the effective domain
\[
\dom u_i:=\{x_i\in\reals^J\,|\,u_i(x_i)>-\infty\}
\]
of $u_i$ is an open set. Both conditions hold if $u_i$ is finite everywhere since concavity then implies that $u_i$ is continuous; see \cite[Corollary~10.1.1]{roc70a}.

\begin{theorem}\label{thm:pareto}
Under Assumption~\ref{ass:g}, Pareto allocations are double auction equilibria. The converse holds under Assumption~\ref{ass:0}.
\end{theorem}

\begin{proof}
Assume that $x^0$ is not a double auction equilibrium. By Lemma~\ref{lem:opt}, there is an $r>0$ such that the constraints of problem \eqref{opt} are satisfied. We can then construct another feasible allocation by giving $rg$ to one of the agents. Under Assumption~\ref{ass:g}, this would lead to strict increase of the agent's utility so $x^0$ can't be Pareto. This proves the first claim.

On the other hand, if $x^0$ is not Pareto, there is a feasible allocation $\tilde x$ such that $u_i(\tilde x_i)\ge u_i(x^0_i)$ for all $i\in I$ and $u_{i'}(\tilde x_{i'})>u_i(x^0_{i'})$ for some $i'\in I$. Under Assumption~\ref{ass:0}, there is an $r>0$ such that $u_{i'}(\tilde x_{i'}-rg)\ge u_{i'}(x^0_{i'})$. Setting
\[
w_i=\begin{cases}
\tilde x_i &\text{for $i\ne i'$},\\
\tilde x_i-rg & \text{for $i=i'$},
\end{cases}
\]
we would then obtain a feasible solution to problem \eqref{opt} with strictly positive optimum value. Thus, by Lemma~\ref{lem:opt}, $x^0$ would not be a double auction equilibrium.
\end{proof}

Theorem~\ref{thm:pareto} gives immediate existence results for the existence of double auction equilibria: under Assumption~\ref{ass:g}, a sufficient condition is the existence of Pareto equilibria. Under Assumption~\ref{ass:0} this is also necessary. This should be compared with the more involved proofs and conditions for the existence of classical Walrasian equilibria; see \cite{jrw7} for some of the most general results on the topic.

The notion of double auction equilibrium is closely related also to the {\em valuation equilibrium} introduced in Debreu~\cite{deb54}. The connections rely on the following dual characterizations of double auction equilibria.

\begin{lemma}\label{lem:ve}
Under Assumptions~\ref{ass:g} and \ref{slater}, the following are equivalent
\begin{enumerate}
\item
  $x^0$ is a double auction equilibrium,
\item
  there is a price vector $p$ with $\partial D_i(0)\ni p$ for all $i\in I$.
\item
  there is a price vector $p$ with $g\cdot p=1$ and, for all $i\in I$,
  \[
  p\cdot w_i\ge p\cdot x^0_i
  \]
  for all $w_i$ with $u_i(w_i)\ge u_i(x^0_i)$.
\end{enumerate}
\end{lemma}

\begin{proof}
The equivalence of the first two follows by taking $\bar x=0$ in  Theorem~\ref{kkt}. By definition, $\partial D_i(0)\ni p$ means that
\begin{align*}
&D_i(x_i)\le D_i(0)+p\cdot x_i\quad\forall x_i\in\reals^J\\
  \iff\ &r\le D_i(0)+p\cdot x_i\quad\forall x_i\in\reals^J,\ \forall r\in\reals:\ u_i(x^0_i+x_i-r g)\ge u_i(x^0_i)\\
\iff\ &r\le D_i(0)+p\cdot(w_i-x^0_i+r g)\quad\forall w_i\in\reals^J,\ \forall r\in\reals:\ u_i(w_i)\ge u_i(x^0_i)\\
\iff\ & p\cdot g=1,\quad 0\le D_i(0)+p\cdot(w_i-x^0_i)\quad\forall w_i\in\reals^J:\ u_i(w_i)\ge u_i(x^0_i),
\end{align*}
where $D_i(0)=0$, by Lemma~\ref{lem:usc}. Thus, 2 is equivalent to 3.
\end{proof}

Combining Lemma~\ref{lem:ve} with Theorem~\ref{thm:pareto} gives the following.

\begin{corollary}\label{cor:ve}
Under Assumptions~\ref{ass:g} and \ref{slater}, each Pareto allocation $x$ has an associated price vector $p\in\reals^J$ such that $g\cdot p=1$ and, for all $i\in I$,
\[
p\cdot w_i\ge p\cdot x^0_i
\]
for all $w_i$ with $u_i(w_i)\ge u_i(x^0_i)$. Conversely, under Assumptions~\ref{ass:g}, \ref{slater} and \ref{ass:0}, the existence of such a price vector implies that $x$ is Pareto efficient.
\end{corollary}

The first implication of Corollary~\ref{cor:ve} was given as \cite[Theorem~2]{deb54} under under slightly different assumptions and without the normalization condition $g\cdot p=1$. The seemingly more involved separation argument used in \cite{deb54} is replaced here by Lemma~\ref{lem:ve} which in turn is based on Theorem~\ref{kkt}, the proof of which also relies on separation. The second part of Corollary~\ref{cor:ve} gives the following result given as Theorem~1 in \cite{deb54} under slightly different conditions.

\begin{corollary}\label{cor:ve2}
If Assumption~\ref{ass:0} holds and $x$ is a feasible allocation such that there is a price vector $p$ with $g\cdot p=1$ and
\[
u_i(w_i)\le u_i(x_i)
\]
for all $w_i$ with $p\cdot w_i\le p\cdot x_i$, then $x$ is Pareto efficient.
\end{corollary}

\begin{proof}
The given condition clearly implies condition 3 in Lemma~\ref{lem:ve} so $x$ is a double auction equilibrium. Theorem~\ref{thm:pareto} now implies that $x$ is Pareto.
\end{proof}

While the sufficient condition in Corollary~\ref{cor:ve2} always implies condition 3 in Lemma~\ref{lem:ve}, the remark on page 591 of \cite{deb54} gives conditions under which the conditions are equivalent. In the terminology of \cite{deb54}, point $x$ satisfying the sufficient condition, except for the normalization condition $g\cdot p=1$, in Corollary~\ref{cor:ve2} was said to be a {\em valuation equilibrium}. Clearly, the normalization condition can be achieved by scaling the price vector provided $g\cdot p>0$. 





\section{Convergence to efficient allocations}\label{sec:ra}

In general, an agent's indifference price function $D_i$ as defined by \eqref{idp} depends on her endowment $x^0_i$. After market clearing, her endowment is changed to $x^1_i = x^0_i + \bar x^0_i - (p\cdot\bar x^0_i)g$ so her demand may change too. There is no reason for the new allocation $x^1$ to be a double auction equilibrium, in general.

Assume now that the auction is repeated indefinitely and denote agent $i$'s position after the $t$th auction by $x_i^t$. That is,
\[
x^t_i := x^{t-1}_i+\bar x^t_i-(p^t\cdot \bar x_i^t)g,
\]
where $p^t$ is the market clearing price and $\bar x_i^t$ is agent $i$'s purchase in the $t$th auction. We will show that, under fairly general conditions, the surplus decreases to zero and all cluster points of the allocations are individually rational and Pareto efficient. This seems to be the first mathematical justification of the efficiency of the double auction mechanism. 

The following assumption, where $\uball(r)$ denotes the ball of radius $r$ in the commodity space $\reals^J$, is a slight strengthening of Assumption~\ref{ass:g}.

\begin{assumption}\label{ass:1}
For every $r>0$ there exists a $\delta_i>0$ such that 
\[
u_i(x_i+r g)\ge u_i(x_i) + \delta_ir\quad\forall x_i\in\uball(r).
\]
\end{assumption}

The concavity of $u_i$ implies, by \cite[Theorem~23.1]{roc70a}, that the difference quotient
\[
\frac{u_i(x_i+r g)-u_i(x_i)}{r}
\]
is nonincreasing in $r$ so the inequality in Assumption~\ref{ass:1} implies\[
u_i(x_i+r' g)\ge u_i(x_i) + \delta_ir'
\]
for any $r'\in(0,r]$.

\begin{lemma}
Assumption~\ref{ass:1} is implied by Assumption~\ref{ass:g} if $x_i+rg\in\inte\dom u_i$ for all $x_i\in\dom u_i$ and $r>0$.
\end{lemma}

\begin{proof}
Assume that Assumption~\ref{ass:1} fails for some $r>0$. This means that there is a convergent sequence $x^\nu\to x\in\uball(r)$ such that
\[
\limsup_{\nu\to\infty}\{u_i(x^\nu_i+rg)-u_i(x^\nu_i)\}=0.
\]
By upper semicontinuity of $u_i$, 
\begin{align*}
  0&\ge\limsup_{\nu\to\infty}u_i(x^\nu_i+rg)-\limsup_{\nu\to\infty}u_i(x^\nu_i)\\
  &\ge\limsup_{\nu\to\infty}u_i(x^\nu_i+rg)-u_i(x_i).
\end{align*}
Concavity of $u_i$ implies its continuity on $\inte\dom u_i$ so if $x_i+rg\in\inte\dom u_i$, the last supremum equals $u_i(x_i+rg)$ contradicting Assumption~\ref{ass:g}.
\end{proof}

An allocation $\bar x$ is said to be {\em individually rational} if $u_i(\bar x_i)\ge u_i(x^0_i)$ for all $i\in I$.
  
\begin{theorem}\label{thm:conv}
Under Assumptions~\ref{ass:g}, \ref{ass:rec} and \ref{ass:1}, the sequence $(x^t)$ is bounded, $CS(x^t)$ decreases to zero with $t$,
\[
CS(x^t)\le\frac{1}{t}\sum_{i\in I}\frac{u_i(x_i^t)-u_i(x_i^0)}{\delta_i}
\]
and the cluster points $\bar x$ of $(x^t)_{t=0}^\infty$ are double auction equilibria and individually rational. In particular, $\bar x$ are Pareto efficient under Assumption~\ref{ass:0}.
\end{theorem}

\begin{proof}
Since all iterates are feasible and, by Remark~\ref{rem:paretoimp}, individually rational the sequence is contained in the set
\[
C = \{x\in\reals^{I\times J}\mid \sum_{i\in I}x_i=\sum_{i\in I}x^0_i,\ u_i(x_i)\ge u_i(x_i^0)\ i\in I\}.
\]
By \cite[Theorem~8.4]{roc70a}, $C$ is bounded if and only if its recession cone
\[
C^\infty := \{x\in\reals^{I\times J}\mid \bar x+\alpha x\in C\ \forall\bar x\in C,\ \alpha>0\}
\]
only consists of the origin. By \cite[Corollary~8.3.3]{roc70a} and \cite[Theorem~8.7]{roc70a},
\[
C^\infty = \{x\in\reals^{I\times J}\mid \sum_{i\in I}x_i=0,\ u_i^\infty(x_i)\ge 0\ i\in I\}
\]
so, by Assumption~\ref{ass:rec}, $C$ is bounded. This proves the first claim.

By Lemma~\ref{lem:opt}, the total surplus $CS(x^{t-1})$ of the $t$th auction is the optimum value of
\begin{equation*}
\begin{aligned}
&\maximize\quad & & r\quad\ovr\ r\in\reals,\ w\in\reals^{I\times J}\\
  &\st\quad & & \sum_{i\in I}w_i + rg = \sum_{i\in I}x^{t-1}_i,\\
  & & & u_i(w_i) \ge u_i(x_i^{t-1})\quad i\in I.
\end{aligned}
\end{equation*}
As in \eqref{eq:mcc}, the market clearing condition implies
\[
\sum_{i\in I}x^{t-1}_i = \sum_{i\in I}x^0_i
\]
for all $t$ while, by Remark~\ref{rem:paretoimp}, $u_i(x_i^t)$ are nondecreasing in $t$. Thus, the constraints become more restrictive with $t$ so the optimum value $CS(x^t)$ is nonincreasing.

Denote agent $i$'s surplus in the $t$th auction by $CS_i(x^{t-1}):=D_i(\bar x_i^t)-p^t\cdot\bar x^t_i$ and let $r>0$ be such that $CS_i(x^t)\le r$ and $x_i^t\in\uball(r)$ for all $t$. Writing
\[
x^t_i=x^{t-1}_i+\bar x^t_i-D_i(\bar x^t_i)g + CS_i(x^{t-1})g,
\]
Assumption~\ref{ass:1} gives
\[
u_i(x_i^t)\ge u_i(x^{t-1}_i+\bar x^t_i-D_i(\bar x^t_i)g) + \delta_iCS_i(x^{t-1}) \ge u_i(x_i^{t-1}) + \delta_iCS_i(x^{t-1}).
\]
Adding up over iterations $s=1,\ldots,t$ gives
\[
\delta_i\sum_{s=0}^{t-1}CS_i(x^s)\le u_i(x_i^t)-u_i(x_i^0)
\]
and adding up over agents
\[
\sum_{s=0}^{t-1}CS(x^s) \le \sum_{i\in I}\frac{u_i(x_i^t)-u_i(x_i^0)}{\delta_i}.
\]
Since $CS(x^t)$ is nonincreasing in $t$, the left side is greater than $tCS(x^t)$ so
\[
CS(x^t)\le\frac{1}{t}\sum_{i\in I}\frac{u_i(x_i^t)-u_i(x_i^0)}{\delta_i}.
\]

If $\bar x$ is a cluster point of $(x^t)$, the upper semicontinuity of $u_i$ and the monotonicity of $u_i(x_i^t)$ in $t$ give 
\[
u_i(\bar x_i)\ge\limsup_{t\to\infty} u_i(x_i^t)=\sup_t u_i(x_i^t).
\]
Thus, by Lemma~\ref{lem:opt}, $CS(\bar x)\le CS(x^t)$ for all $t$ so $CS(\bar x)=0$.
\end{proof}

Theorem~\ref{thm:conv} restates the fundamental fact of welfare economics that ``competitive markets'' lead to efficient allocations. While the classical Walrasian model of the market leads to an equilibrium in a single trade, it assumes that the equilibrium prices are given exogenously or through a tatonnement process which is at odds with existing market mechanisms. Our market model gives a more realistic description of markets where prices are formed endogenously and the auction is repeated in order to reach equilibrium. Moreover, our result gives a worst-case bound on the speed of convergence which has been observed in extensive empirical studies ever since the pioneering works of Smith~\cite{smi62}.

\section{A numerical illustration}\label{sec:num}

This section presents a numerical study of the double auction process studied in the previous section. In the example, we will assume Cobb-Douglas utilities
\[
u_i(x_i)=\prod_{j\in J}x_{i,j}^{\alpha_{i,j}},
\]
where $\alpha_{i,j}$ are positive parameters with
\[
\sum_{j\in J}\alpha_{i,j}=1.
\]
For each agent $i\in I$, we generate the parameter vector $\alpha_i=(\alpha_{i,j})_{j\in J}$ randomly by drawing a vector from the uniform distribution over the unit cube and then scaling the vector to the unit simplex. The initial endowments $x^0_i$ are randomly drawn from the uniform distribution over the unit cube.

We assume that all prices are quoted in terms of cash which is assumed to be asset $0\in J$. This means that the numeraire portfolio $g$ is the unit vector with $g_0=1$ and $g_j=0$ for $j\in J\setminus\{0\}$. At each iteration of the double auction, we solve the market clearing problem \eqref{p} by solving the equivalent problem \eqref{opt} in Lemma~\ref{lem:opt} with the conic interior point solver of MOSEK~\cite{mosek}. The problem is formulated and communicated to MOSEK using Python~3.7 and CVXPY~\cite{cvxpy}.

Table~\ref{tab:num} illustrates the progress of the iterated double auctions  with $100$ agents and $5$ assets. The first column is the iteration counter, the second column is the total consumer surplus, the third column is the sum of logarithmic utilities, the fourth column is the inner product between the total endowment
\[
e=\sum_{i\in I} x^0_i
\]
and the market clearing prices $p^t$, the fifth column gives the Euclidean norm of the portfolio updates
\[
\Delta x^t :=x^t-x^{t-1}
\]
and the last column gives the market clearing prices $p^t$. The auction is iterated until the total consumer surplus falls below $0.001$. The computation time to generate the data, to set up the problem and to run the $17$ iterations was $1.20$ seconds using Intel Core i7-1065G7 processor with 15.2GB of RAM running Linux.

We repeat the experiment with the same endowments and utility functions but with numeraire portfolio $g=(1,\ldots,1)$ instead of the unit vector $g=(1,0,\ldots,0)$. The results are given in Table~\ref{tab:num2}. The convergence seems much faster this time. With the new numeraire and the Cobb-Douglass utilities, the constants $\delta_i$ in Assumption~\ref{ass:1} can be taken larger. According to Theorem~\ref{thm:conv}, this gives a better complexity bound which explains the speedup. It is also interesting to noticee that, with numeraire portfolio $g-(1,\ldots,1)$, the sum of the log-utilities is higher than in the previous example already after the first iteration. This illustrates the fact that, while the iterates are individually rational and the accumulation points are Pareto efficient, they are not unique.

\begin{table}\label{tab:num}
\centering
{\tiny
\begin{tabular}{l|c|c|c|ccccccc}
\toprule
$t$ &     CS &       $\sum_i\ln u_i$ &   $\|\Delta x^t\|$ &     &     &    $p_t$ &     &     \\
\midrule
1  & 38.691 & -68.341 & 8.705 & 1.000 & 0.159 & 0.135 & 0.132 & 0.132 \\
2  & 11.534 & -63.125 & 2.781 & 1.000 & 0.690 & 0.585 & 0.580 & 0.577 \\
3  &  4.266 & -61.498 & 1.085 & 1.000 & 0.843 & 0.711 & 0.710 & 0.709 \\
4  &  2.091 & -60.721 & 0.559 & 1.000 & 0.894 & 0.754 & 0.753 & 0.753 \\
5  &  1.159 & -60.310 & 0.329 & 1.000 & 0.916 & 0.774 & 0.772 & 0.772 \\
6  &  0.711 & -60.084 & 0.186 & 1.000 & 0.928 & 0.784 & 0.782 & 0.782 \\
7  &  0.486 & -59.951 & 0.119 & 1.000 & 0.933 & 0.790 & 0.787 & 0.786 \\
8  &  0.362 & -59.864 & 0.080 & 1.000 & 0.936 & 0.792 & 0.789 & 0.789 \\
9  &  0.289 & -59.796 & 0.063 & 1.000 & 0.938 & 0.794 & 0.791 & 0.791 \\
10 &  0.236 & -59.735 & 0.060 & 1.000 & 0.939 & 0.795 & 0.792 & 0.792 \\
11 &  0.187 & -59.682 & 0.058 & 1.000 & 0.940 & 0.796 & 0.794 & 0.793 \\
12 &  0.141 & -59.635 & 0.055 & 1.000 & 0.941 & 0.797 & 0.795 & 0.795 \\
13 &  0.098 & -59.595 & 0.051 & 1.000 & 0.942 & 0.798 & 0.796 & 0.796 \\
14 &  0.061 & -59.565 & 0.044 & 1.000 & 0.943 & 0.799 & 0.797 & 0.797 \\
15 &  0.032 & -59.545 & 0.035 & 1.000 & 0.944 & 0.800 & 0.797 & 0.797 \\
16 &  0.012 & -59.536 & 0.024 & 1.000 & 0.944 & 0.800 & 0.798 & 0.798 \\
17 &  0.002 & -59.534 & 0.011 & 1.000 & 0.944 & 0.800 & 0.798 & 0.798 \\
\end{tabular}
}
\caption{Consumer surplus, sum of log-utilities, Euclidean norm of the updated positions, and the market clearing prices along iterates when the numeraire portfolio is $g=(1,0,\ldots,0)$}
\end{table}

\begin{table}\label{tab:num2}
\centering
{\tiny
\begin{tabular}{l|c|c|c|ccccccc}
\toprule
$t$ &     CS &       $\sum_i\ln u_i$ &   $\|\Delta x^t\|$ &     &     &    $p_t$ &     &     \\
\midrule
1 & 15.894 & -58.257 & 7.278 & 0.235 & 0.224 & 0.187 & 0.177 & 0.177 \\
2 &  1.905 & -54.219 & 2.181 & 0.230 & 0.220 & 0.184 & 0.183 & 0.182 \\
3 &  0.206 & -53.811 & 0.537 & 0.231 & 0.219 & 0.184 & 0.184 & 0.182 \\
4 &  0.027 & -53.759 & 0.140 & 0.231 & 0.219 & 0.184 & 0.184 & 0.182 \\
5 &  0.003 & -53.754 & 0.031 & 0.231 & 0.219 & 0.184 & 0.184 & 0.182 \\
\end{tabular}
}
\caption{Consumer surplus, sum of log-utilities, Euclidean norm of the updated positions, and the market clearing prices along iterates when the numeraire portfolio is $g=(1,\ldots,1)$}
\end{table}


Theorem~\ref{thm:conv} and the numerical results in this section are concerned with double auction markets with sealed nonstrategic bidding among agents whose preferences remain fixed. Fixed preferences may be justified over short periods of time but as soon as relevant news and other information become available to the agents, their preferences may change. This may break the double auction equilibrium inducing trades and changes in market clearing prices. Our model could be used to describe such dynamics by allowing the utility functions to change over time. Their dependence on the agents' subjective information could be modeled e.g.\ by the construction in Example~\ref{ex:idu} where the agents' subjective probability measures $P_i$ could be updated with the arrival of new information. The supremum in the definition would amount to adding an extra set of variables in the formulation \eqref{opt} in Lemma~\ref{lem:opt} of the market clearing problem.

Another topic that deserves attention is the relaxation of the assumption of nonstrategic bidding. Even in the usual call auctions where all bids are sealed, one could consider agents who deviate from the competitive bidding described by the indifference rule \eqref{idp} in the hope of gaining strategic advantage as the auction gets repeated. One could approach such situations with game theoretic formulations and/or probabilistic modeling of the other bidders like in Myerson~\cite{mye81} in the optimal design of one-sided auctions. Such models would require different kind of techniques well beyond the scope of the present paper. Empirical research suggest, however, that the gains from strategic bidding may be modest; see e.g.\ \cite{fo95,cf96,zf7}.



\bibliographystyle{plain}
\bibliography{sp}

\end{document}